\documentclass[submission,copyright,creativecommons]{eptcs}
\usepackage{amssymb,amsmath,mathpartir,paralist,amsthm,xspace}
\usepackage[utf8]{inputenc}
\usepackage[textsize=small]{todonotes}
\usepackage[T1]{fontenc}

\newcommand{\Nat}{\mathbb{N}}
\newcommand{\lY}{$\lambda Y$}
\newcommand{\BT}{\mathit{BT}}
\newcommand{\br}{\mathsf{br}}
\newcommand{\arr}{{\to}}
\newcommand{\bred}{\to_\beta}
\newcommand{\restr}{{\restriction}}
\newcommand{\ord}{\mathit{ord}}
\newcommand{\Split}{\mathit{Split}}
\newcommand{\Comp}{\mathit{Comp}}
\newcommand{\palka}{\mid}
\newcommand{\palkaC}{;\allowbreak}

\newcommand{\AppR}{\TirName{(\!@\!)}\xspace}
\newcommand{\LamR}{\TirName{($\lambda$)}\xspace}
\newcommand{\ConR}{\TirName{(Con)}\xspace}
\newcommand{\VarR}{\TirName{(Var)}\xspace}
\newcommand{\Gammae}{\varepsilon}

\newcommand{\Dd}{\mathcal{D}}
\newcommand{\Gg}{\mathcal{G}}
\newcommand{\Ll}{\mathcal{L}}
\newcommand{\Pp}{\mathcal{P}}
\newcommand{\Tt}{\mathcal{T}}
\newcommand{\Ttrip}{\mathcal{F}}

\newtheorem{theorem}{Theorem}
\newtheorem{lemma}[theorem]{Lemma}

\theoremstyle{definition}
\newtheorem{example}{Example}

\providecommand{\event}{ITRS 2016}
\title{Intersection Types and Counting%
	\thanks{Work supported by the National Science Center (decision DEC-2012/07/D/ST6/02443).
	A full version of this paper is available \cite{itrs-arxiv}}}
\author{Paweł Parys
	\institute{University of Warsaw, Poland}
	\email{parys@mimuw.edu.pl}
}
\def\titlerunning{Intersection Types and Counting}
\def\authorrunning{P. Parys}
\begin{document}
\maketitle

\begin{abstract}
	We present a new approach to the following meta-problem: given a quantitative property of trees,
	design a type system such that the desired property for the tree generated by an infinitary ground $\lambda$-term corresponds to some property of a derivation of a type for this $\lambda$-term, in this type system.

	Our approach is presented in the particular case of the language finiteness problem for nondeterministic higher-order recursion schemes (HORSes):
	given a nondeterministic HORS, decide whether the set of all finite trees generated by this HORS is finite.
	We give a type system such that the HORS can generate a tree of an arbitrarily large finite size if and only if in the type system we can obtain derivations that are arbitrarily large, in an appropriate sense;
	the latter condition can be easily decided.
\end{abstract}

\section{Introduction}
	In this paper we consider \lY-calculus, which is an extension of the simply typed $\lambda$-calculus by a fixed-point operator $Y$. 
	A term $P$ of \lY-calculus that is of sort\footnote{%
		We use the word ``sort'' instead of the usual ``type'' to avoid confusion with intersection types introduced in this paper.}
	$o$ can be used to generate an infinite tree $\BT(P)$, called the B\"ohm tree of $P$.
	Trees generated by terms of \lY-calculus can be used to faithfully represent the control flow of programs in languages with higher-order functions.
	Traditionally, Higher Order Recursive Schemes (HORSes) are used for this purpose~\cite{Damm82,KNU-hopda,Ong-hoschemes,kobayashiOng2009type}; this formalism is equivalent to \lY-calculus,
	and the translation between them is rather straightforward~\cite{MSC:9613738}.
	Collapsible Pushdown Systems \cite{collapsible} and Ordered Tree-Pushdown Systems \cite{DBLP:conf/fsttcs/ClementePSW15} are other equivalent formalisms.

	Intersection type systems were intensively used in the context of HORSes, for several purposes like
	model-checking \cite{kobayashi2009types-popl,kobayashiOng2009type,DBLP:conf/csl/BroadbentK13,DBLP:conf/popl/RamsayNO14},
	pumping \cite{koba-pumping},
	transformations of HORSes \cite{context-sensitive-2,downward-closure}, etc.
	Interestingly, constructions very similar to intersection types were used also on the side of collapsible pushdown systems; they were alternating stack automata \cite{saturation}, and types of stacks \cite{ho-new,Kar-Par-pumping}.

	In this paper we show how intersection types can be used for deciding quantitative properties of trees generated by \lY-terms.
	We concentrate on the language finiteness problem for nondeterministic HORSes:
	given a nondeterministic HORS, decide whether the set of all finite trees generated by this HORS is finite.
	
	This problem can be restated in the world of \lY-terms (or standard, deterministic HORSes), generating a single infinite tree.
	Here, instead of resolving nondeterministic choices during the generation process, we leave them in the resulting tree.
	Those nondeterministic choices are denoted by a distinguished $\br$ (``branch'') symbol, below which we put options that could be chosen.
	Then to obtain a finite tree generated by the original HORS we just need to recursively choose in every $\br$-labeled node which of the two subtrees we want to consider.
	Thus, in this setting, the language finiteness problem asks whether the set of all finite trees obtained this way is finite.
	
	The difficulty of this problem lies in the fact that sometimes the same finite tree may be found in infinitely many different places of $\BT(P)$ (i.e., generated by a nondeterministic HORS in many ways);
	thus the actual property to decide is whether there is a common bound on the size of each of these trees.
	This makes the problem inaccessible for standard methods used for analyzing HORSes, as they usually concern only regular properties of the B\"ohm tree, while boundedness is a problem of different kind.
	The same difficulty was observed in \cite{koba-pumping}, where they prove a pumping lemma for deterministic HORSes, while admitting (Remark 2.2) that their method is too weak to reason about nondeterministic HORSes.

	In order to solve the language finiteness problem, we present an appropriate intersection type system, where derivations are annotated by flags and markers of multiple kinds.
	The key property of this type system is that the number of flags in a type derivation for a \lY-term $P$ 
	approximates the size of some finite tree obtained by resolving nondeterministic choices in the infinite tree $\BT(P)$.
	In consequence, there are type derivations using arbitrarily many flags if, and only if, the answer to the language finiteness problem is ``no''.

	The language finiteness problem was first attacked in \cite{achim-pumping} (for safe HORSes only), but their algorithm turned out to be incorrect \cite{achim-erratum}.
	To our knowledge, the only known solution of this problem follows from a recent decidability result for the diagonal problem \cite{diagonal-safe, downward-closure}.
	This problem asks, given a nondeterministic HORS and a set of letters $\Sigma$, whether for every $n\in\Nat$ the HORS generates a finite tree in which every letter from $\Sigma$ appears at least $n$ times.
	Clearly, a nondeterministic HORS generates arbitrarily large trees exactly when for some letter $a$ it generates trees having arbitrarily many $a$ letters, i.e., when the answer to the diagonal problem for $\Sigma=\{a\}$ is ``yes''.

	Our type system is, to some extent, motivated by the algorithm of \cite{downward-closure} solving the diagonal problem.
	This algorithm works by repeating two kinds of transformations of HORSes.
	The first of them turns the HORS into a HORS generating trees having only a fixed number of branches, one per each letter from $\Sigma$ (i.e., one branch in our case of $|\Sigma|=1$).
	The branches are chosen nondeterministically out of some tree generated by the original HORS; for every $a\in\Sigma$ there is a choice witnessing that $a$ appeared many times in the original tree.	
	Then such a HORS of the special form is turned into a HORS that is of order lower by one,
	and generates trees having the same nodes as trees generated by the original HORS, but arranged differently (in particular, the new trees may have again arbitrarily many branches).
	After finitely many repetitions of this procedure, a HORS of order $0$ is obtained, and the diagonal problem becomes easily decidable.
	In some sense we want to do the same, but instead of applying all these transformations one by one, we simulate all of them simultaneously in a single type derivation.
	In this derivation, for each order $n$, we allow to place arbitrarily one marker ``of order $n$''; this corresponds to the nondeterministic choice of one branch in the $n$-th step of the previous algorithm.
	We also place some flags ``of order $n$'', in places that correspond to nodes remaining after the $n$-th step of the previous algorithm.
	
	The idea of using intersection types for counting is not completely new.
	Paper \cite{jfp-numerals} presents a type system that, essentially, allows to estimate the size of the $\beta$-normal form of a $\lambda$-term just by looking at (the number of some flags in) a derivation of a type for this term.
	A similar idea, but for higher-order pushdown automata, is present in \cite{ho-new}, where we can estimate the number of $\sharp$ symbols appearing on a particular, deterministically chosen branch of the generated tree.
	This previous approach also uses intersection types, where the derivations are marked with just one kind of flags, denoting ``productive'' places of a $\lambda$-term
	(oppositely to our approach, where we have different flags for different orders, and we also have markers).
	The trouble with the ``one-flag'' approach is that it works well only in a completely deterministic setting, where looking independently at each node of the B\"ohm tree we know how it contributes to the result;
	the method stops working (or at least we do not know how to prove that it works) in our situation, where we first nondeterministically perform some guesses in the B\"ohm tree, and only after that we want to count something that depends on the chosen values.

	\paragraph{Acknowledgements.} 
	I would like to thank Szymon Toruńczyk for stimulating discussions, and anonymous reviewers for useful comments.

\section{Preliminaries}

	\paragraph{Trees.}
	Let $\Sigma$ be a \emph{ranked alphabet}, i.e., a set of symbols together with a rank function assigning a nonnegative integer to each of the symbols.
	We assume that $\Sigma$ contains a distinguished symbol $\br$ of rank $2$, used to denote nondeterministic choices.
	A \emph{$\Sigma$-labeled} tree is a tree that is rooted (there is a distinguished root node),
	node-labeled (every node has a label from $\Sigma$), 
	ranked (a node with label of rank $n$ has exactly $n$ children), 
	and ordered (children of a node of rank $n$ are numbered from $1$ to $n$).

	When $t$ is a $\Sigma$-labeled tree $t$, by $\Ll(t)$ we denote the set of all finite trees that can be obtaining by choosing in every $\br$-labeled node of $t$ which of the two subtrees we want to consider.
        More formally, we consider the following relation $\to_\br$: we have $t\to_\br u$ if $u$ can be obtained from $t$ by choosing in $t$ a $\br$-labeled node $x$ and its child $y$, 
        and replacing the subtree starting in $x$ by the subtree starting in $y$ (which removes $x$ and the other subtree of $x$).
        Let $\to_\br^*$ be the reflexive transitive closure of $\to_\br$.
        Then $\Ll(t)$ contains all trees $u$ that do not use the $\br$ label, are finite, and such that $t\to_\br^*u$.

	\paragraph{Infinitary $\lambda$-calculus.}
	The set of \emph{sorts} (a.k.a.~simple types), constructed from a unique basic sort $o$ using a binary operation $\arr$, is defined as usual.
	The order of a sort is defined by: $\ord(o)=0$, and $\ord(\alpha\arr\beta)=\max(1+\ord(\alpha),\ord(\beta))$.
	
	We consider infinitary, sorted $\lambda$-calculus. 
	\emph{Infinitary $\lambda$-terms} (or just \emph{$\lambda$-terms}) are defined by coinduction, according to the following rules:
	\begin{itemize}
		\item	if $a\in\Sigma$ is a symbol of rank $r$, and $P_1^o,\dots,P_r^o$ are $\lambda$-terms, then $(a\,P_1^o\,\dots\,P_r^o)^o$ is a $\lambda$-term,
		\item	for every sort $\alpha$ there are infinitely many variables $x^\alpha,y^\alpha,z^\alpha,\dots$; each of them is a $\lambda$-term,
		\item	if $P^{\alpha\arr\beta}$ and $Q^\alpha$ are $\lambda$-terms, then $(P^{\alpha\arr\beta}\,Q^\alpha)^\beta$ is a $\lambda$-term, and 
		\item	if $P^\beta$ is a $\lambda$-term and $x^\alpha$ is a variable, then $(\lambda x^\alpha.P^\beta)^{\alpha\arr\beta}$ is a $\lambda$-term.
	\end{itemize}
	We naturally identify $\lambda$-terms differing only in names of bound variables.
	We often omit the sort annotations of $\lambda$-terms, but we keep in mind that every $\lambda$-term (and every variable) has a particular sort.
	A $\lambda$-term $P$ is \emph{closed} if it has no free variables.
	Notice that, for technical convenience, a symbol of positive rank is not a $\lambda$-term itself, but always comes with arguments.
	This is not a restriction, since e.g.~instead of a unary symbol $a$ one may use the term $\lambda x.a\,x$.

	The order of a $\lambda$-term is just the order of its sort.
	The \emph{complexity} of a $\lambda$-term $P$ is the smallest number $m$ such that the order of every subterm of $P$ is at most $m$.
	We restrict ourselves to $\lambda$-terms that have finite complexity.
	
	A $\beta$-reduction is defined as usual.
	We say that a $\beta$-reduction $P\bred Q$ \emph{is of order $n$} if it concerns a redex $(\lambda x.R)\,S$ such that $\ord(\lambda x.R)=n$.
	In this situation the order of $x$ is at most $n-1$, but may be smaller (when other arguments of $R$ are of order $n-1$).

	\paragraph{B\"ohm Trees.}

	We consider B\"ohm trees only for closed $\lambda$-terms of sort $o$.
	For such a term $P$, its \emph{B\"ohm tree} $\BT(P)$ is constructed by coinduction, as follows:
	if there is a sequence of $\beta$-reductions from $P$ to a $\lambda$-term of the form $a\,P_1\,\ldots\,P_r$ (where $a$ is a symbol),
	then the root of the tree $t$ has label $a$ and $r$ children, and the subtree starting in the $i$-th child is $\BT(P_i)$.
	If there is no sequence of $\beta$-reductions from $P$ to a $\lambda$-term of the above form, then $\BT(P)$ is the full binary tree with all nodes labeled by $\br$.\footnote{%
		Usually one uses a special label $\bot$ of rank $0$ for this purpose, but from the perspective of our problem both definitions are equivalent.}
	By $\Ll(P)$ we denote $\Ll(\BT(P))$.

	\paragraph{\lY-calculus.}

	The syntax of \lY-calculus is the same as that of finite $\lambda$-calculus, extended by symbols $Y^{(\alpha\arr\alpha)\arr\alpha}$, for each sort $\alpha$.
	A term of \lY-calculus is seen as a term of infinitary $\lambda$-calculus 
	if we replace each symbol $Y^{(\alpha\arr\alpha)\arr\alpha}$ by the unique infinite $\lambda$-term $Z$ such that $Z$ is syntactically the same as $\lambda x^{\alpha\arr\alpha}.x\,(Z\,x)$.
	In this way, we view \lY-calculus as a fragment of infinitary $\lambda$-calculus.
 
	It is standard to convert a nondeterministic HORS $\Gg$ into a closed \lY-term $P^o$ such that $\Ll(P)$ is exactly the set of all finite trees generated by $\Gg$.
	The following theorem, which is our main result, states that the \emph{language finiteness problem} is decidable.
	
	\begin{theorem}\label{thm:main}
		Given a closed \lY-term $P$ of sort $o$, one can decide whether $\Ll(P)$ is finite.
	\end{theorem}

\section{Intersection Type System}

	In this section we introduce a type system that allows to determine the desired property: whether in $\Ll(P)$ there is an arbitrarily large tree.

	\paragraph{Intuitions.}\label{para:intuitions}

	The main novelty of our type system is in using flags and markers, which may label nodes of derivation trees.
	To every flag and marker we assign a number, called an order.
	While deriving a type for a $\lambda$-term of complexity $m$, we may place in every derivation tree at most one marker of each order $n\in\{0,\dots,m-1\}$, and arbitrarily many flags of each order $n\in\{0,\dots,m\}$.
	
	Consider first a $\lambda$-term $M_0$ of complexity $0$.
	Such a term actually equals its B\"ohm tree.
	Our aim is to describe some finite tree $t$ in $\Ll(M_0)$, i.e., obtained from $M_0$ by resolving nondeterministic choices in some way.
	We thus just put flags of order $0$ in all those (appearances of) symbols in $M_0$ that contribute to this tree $t$;
	the type system ensures that indeed all symbols of some finite tree in $\Ll(M_0)$ are labeled by a flag.
	Then clearly we have the desired property that there is a derivation with arbitrarily many flags if, and only if, there are arbitrarily large trees in $\Ll(M_0)$.
	
	Next, consider a $\lambda$-term $M_1$ that is of complexity $1$, and reduces to $M_0$.
	Of course every finite tree from $\Ll(M_0)$ is composed of symbols appearing already in $M_1$;
	we can thus already in $M_1$ label (by order-$0$ flags) all symbols that contribute to some tree $t\in\Ll(M_0)$ (and an intersection type system can easily check correctness of such labeling).
	There is, however, one problem: a single appearance of a symbol in $M_1$ may result in many appearances in $M_0$ (since a function may use its argument many times).
	Due to this, the number of order-$0$ flags in $M_1$ does not correspond to the size of $t$.
	We rescue ourselves in the following way.
	In $t$ we choose one leaf, we label it by an order-$0$ marker, and on the path leading from the root to this marker we place order-$1$ flags.
	On the one hand, $\Ll(M_0)$ contains arbitrarily large trees if, and only if, it contains trees with arbitrarily long paths, i.e., trees with arbitrarily many order-$1$ flags.
	On the other hand, we can perform the whole labeling (and the type system can check its correctness) already in $M_1$, and the number of order-$1$ flags in $M_1$ will be precisely the same as it would be in $M_0$.
	Indeed, in $M_1$ we have only order-$1$ functions, i.e., functions that take trees and use them as subtrees of larger trees;
	although a tree coming as an argument may be duplicated, the order-$0$ marker can be placed in at most one copy.
	This means that, while reducing $M_1$ to $M_0$, every symbol of $M_1$ can result in at most one symbol of $M_0$ lying on the selected path to the order-$0$ marker
	(beside of arbitrarily many symbols outside of this path).
	
	This procedure can be repeated for $M_2$ of complexity $2$ that reduces to $M_1$ via $\beta$-reductions of order $2$ (and so on for higher orders).
	We now place a marker of order $1$ in some leaf of $M_1$;
	afterwards, we place an order-$2$ flag in every node that is on the path to the marked leaf, and that has a child outside of this path whose some descendant is labeled by an order-$1$ flag.
	In effect, for some choice of a leaf to be marked, the number of order-$2$ flags approximates the number of order-$1$ flags, up to logarithm.
	Moreover, the whole labeling can be done in $M_2$ instead of in $M_1$, without changing the number of order-$2$ flags.
	
	In this intuitive description we have talked about labeling ``nodes of a $\lambda$-term'', but formally we label nodes of a derivation tree that derives a type for the term, in our type system.
	Every such node contains a type judgment for some subterm of the term.

	\paragraph{Type Judgments.}
	
	For every sort $\alpha$ we define the set $\Tt^\alpha$ of \emph{types} of sort $\alpha$,
	and the set $\Ttrip^\alpha$ of \emph{full types} of sort $\alpha$.
	This is done as follows, where $\Pp$ denotes the powerset:
	\begin{align*}
		&\Tt^{\alpha\arr\beta}=\Pp(\Ttrip_{\ord(\alpha\arr\beta)}^\alpha)\times\Tt^\beta\,,\qquad
		\Tt^o=o\,,\\
		&\Ttrip_k^\alpha=\{(k,F,M,\tau)\mid F,M\subseteq\{0,\dots,k-1\},\,F\cap M=\emptyset,\,\tau\in\Tt^\alpha\}\,,\qquad\Ttrip^\alpha=\bigcup_{k\in\Nat}\Ttrip_k^\alpha\,.
	\end{align*}
	Notice that the sets $\Tt^\alpha$ and $\Ttrip_k^\alpha$ are finite (unlike $\Ttrip^\alpha$).
	A type $(T,\tau)\in\Tt^{\alpha\arr\beta}$ is denoted as $T\arr\tau$.
	A full type $\hat\tau=(k,F,M,\tau)\in\Ttrip_k^\alpha$ consists of its order $k$, a set $F$ of flag orders, a set $M$ of marker orders, and a type $\tau$;
	we write $\ord(\hat\tau)=k$.
	In order to distinguish types from full types, the latter are denoted by letters with a hat, like $\hat\tau$.
	
	A \emph{type judgment} is of the form $\Gamma\vdash P:\hat\tau\triangleright c$, where $\Gamma$, called a \emph{type environment}, 
	is a function that maps every variable $x^\alpha$ to a subset of $\Ttrip^\alpha$,
	$P$ is a $\lambda$-term, $\hat\tau$ is a full type of the same sort as $P$ (i.e., $\hat\tau\in\Ttrip^\beta$ when $P$ is of sort $\beta$), and $c\in\Nat$.
	
	As usual for intersection types, the intuitive meaning of a type $T\arr\tau$ is that a $\lambda$-term having this type can return a $\lambda$-term having type $\tau$, while taking an argument for which we can derive all full types from $T$.
	Moreover, in $\Tt^o$ there is just one type $o$, which can be assigned to every $\lambda$-term of sort $o$.
	Suppose that we have derived a type judgment $\Gamma\vdash P:\hat\tau\triangleright c$ with $\hat\tau=(m,F,M,\tau)$.
	Then
	\begin{itemize}
	\item	$\tau$ is the type derived for $P$;
	\item	$\Gamma$ contains full types that could be used for free variables of $P$ in the derivation;
	\item	$m$ bounds the order of flags and markers that could be used in the derivation: flags could be of order at most $m$, and markers of order at most $m-1$;
	\item	$M\subseteq\{0,\dots,m-1\}$ contains the orders of markers used in the derivation, together with those provided by free variables 
		(i.e., we imagine that some derivations, specified by the type environment, are already substituted in our derivation for free variables);
		we, however, do not include markers provided by arguments of the term (i.e., coming from the sets $T_i$ when $\tau=T_1\arr\dots\arr T_k\arr o$);
	\item	$F$ contains those numbers $n\in\{0,\dots,m-1\}$ (excluding $n=m$) for which a flag of order $n$ is placed in the derivation itself, or provided by a free variable, or provided by an argument;
		for technical convenience we, however, remove $n$ from $F$ whenever $n\in M$ 
		(when $n\in M$, the information about order-$n$ flags results in placing an order-$(n+1)$ flag, and need not to be further propagated);
	\item	$c$, called a \emph{flag counter}, counts the number of order-$m$ flags present in the derivation.
	\end{itemize}

	\paragraph{Type System.}
	
	Before giving rules of the type system, we need a few definitions.
	We use the symbol $\uplus$ to denote disjoint union.
	When $A\subseteq\Nat$ and $n\in\Nat$, we write $A\restr_{<n}$ for $\{k\in A\mid k<n\}$, and similarly $A\restr_{\geq n}$ for $\{k\in A\mid k\geq n\}$.
	By $\Gammae$ we denote the type environment mapping every variable to $\emptyset$,
	and by $\Gamma[x\mapsto T]$ the type environment mapping $x$ to $T$ and every other variable $y$ to $\Gamma(y)$.

	Let us now say how a type environment $\Gamma$ from the conclusion of a rule may be split into type environments $(\Gamma_i)_{i\in I}$ used in premisses of the rule:
	we say that $\Split(\Gamma\palka(\Gamma_i)_{i\in I})$ holds if and only if for every variable $x$ it holds $\Gamma_i(x)\subseteq\Gamma(x)$ for every $i\in I$,
	and every full type from $\Gamma(x)$ providing some markers (i.e., $(k,F,M,\tau)$ with $M\neq\emptyset$) appears in some $\Gamma_i(x)$.
	Full types with empty $M$ may be discarded and duplicated freely.
	This definition forbids to discard full types with nonempty $M$, and from elsewhere it will follow that they cannot be duplicated.
	As a special case $\Split(\Gamma\palka\Gamma')$ describes how a type environment can be weakened.

	All type derivations are assumed to be finite
	(although we derive types mostly for infinite $\lambda$-terms, each type derivation analyzes only a finite part of a term).
	Rules of the type system will guarantee that the order $m$ of derived full types will be the same in the whole derivation (although in type environments there may be full types of different orders).
	
	We are ready to give the first three rules of our type system:
	\begin{mathpar}
		\inferrule*[right=(Br)]{
			\Gamma\vdash P_i:\hat\tau\triangleright c
			\\
			i\in\{1,2\}
			}
			{\Gamma\vdash \br\,P_1\,P_2:\hat\tau\triangleright c}
		\and
		\inferrule*[right=(Var)]{
			\Split(\Gamma\palka\Gammae[x\mapsto\{(k,F,M',\tau)\}])
			\\
			M\restr_{<k}=M'
			}
			{\Gamma\vdash x:(m,F,M,\tau)\triangleright 0} 
	\end{mathpar}
	\begin{mathpar}
		\inferrule*[right=($\lambda$)]{\Gamma'[x\mapsto T]\vdash P:(m,F,M,\tau)\triangleright c
			\\
			\Split(\Gamma\palka\Gamma')
			\\
			\Gamma'(x)=\emptyset}
	        	{\Gamma\vdash\lambda x.P:(m,F,M\setminus\bigcup{}_{(k,F',M',\sigma)\in T}M',T\arr\tau)\triangleright c}
	\end{mathpar}
	
	We see that to derive a type for the nondeterministic choice $\br\,P_1\,P_2$, we need to derive it either for $P_1$ or for $P_2$.

	The \VarR rule allows to have in the resulting set $M$ some numbers that do not come from the set $M'$ assigned to $x$ by the type environment; these are the orders of markers placed in the leaf using this rule.
	Notice, however, that we allow here only orders not smaller than $k$ (which is the order of the superterm $\lambda x.P$ binding this variable $x$).
	This is consistent with the intuitive description of the type system (page \pageref{para:intuitions}), 
	which says that a marker of order $n$ can be put in a place that will be a leaf after performing all $\beta$-reductions of orders greater than $n$.
	Indeed, the variable $x$ remains a leaf after performing $\beta$-reductions of orders greater than $k$, but while performing $\beta$-reductions of order $k$ this leaf will be replaced by a subterm substituted for $x$.
	Recall also that, by definition of a type judgment, we require that $(k,F,M',\tau)\in\Ttrip^\alpha_k$ and $(m,F,M,\tau)\in\Ttrip^\alpha_m$, for appropriate sort $\alpha$;
	this introduces a bound on maximal numbers that may appear in the sets $F$ and $M$.
	
	\begin{example}\label{ex:var}
		Denoting $\hat\rho_1=(1,\emptyset,\{0\},o)$ we can derive:
		\begin{mathpar}
			\inferrule*[Right=(Var)]{ }{
				\Gammae[x\mapsto\{\hat\rho_1\}]\vdash x:(2,\emptyset,\{0\},o)\triangleright 0
			}
			\and
			\inferrule*[Right=(Var)]{ }{
				\Gammae[x\mapsto\{\hat\rho_1\}]\vdash x:(2,\emptyset,\{0,1\},o)\triangleright 0
			}
		\end{mathpar}
		In the derivation on the right, the marker of order $1$ is placed in the conclusion of the rule.
	\end{example}

	The \LamR rule allows to use (in a subderivation concerning the $\lambda$-term $P$) the variable $x$ with all full types given in the set $T$.
	When the sort of $\lambda x.P$ is $\alpha\arr\beta$, by definition of $\Tt^{\alpha\arr\beta}$ we have that all full types in $T$ have the same order $k=\ord(\alpha\arr\beta)$ (since $(T\arr\tau)\in\Tt^{\alpha\arr\beta}$).
	Recall that we intend to store in the set $M$ the markers contained in the derivation itself and those provided by free variables, but not those provided by arguments.
	Because of this, in the conclusion of the rule we remove from $M$ the markers provided by $x$.
	This operation makes sense only because there is at most one marker of each order, so markers provided by $x$ cannot be provided by any other free variable nor placed in the derivation itself.
	The set $F$, unlike $M$, stores also flags provided by arguments, so we do not need to remove anything from $F$.
	
	\begin{example}\label{ex:lambda}
		The \LamR rule can be used, e.g., in the following way (where $a$ is a symbol of rank $1$):
		\begin{mathpar}
			\inferrule*[Right=($\lambda$)]{
				\Gammae[x\mapsto\{\hat\rho_1\}]\vdash a\,x:(2,\{1\},\{0\},o)\triangleright 0
			}{
				\Gammae\vdash\lambda x.a\,x:(2,\{1\},\emptyset,\{\hat\rho_1\}\arr o)\triangleright 0
			}
			\and
			\inferrule*[Right=($\lambda$)]{
				\Gammae[x\mapsto\{\hat\rho_1\}]\vdash a\,x:(2,\emptyset,\{0,1\},o)\triangleright 1
			}{
				\Gammae\vdash\lambda x.a\,x:(2,\emptyset,\{1\},\{\hat\rho_1\}\arr o)\triangleright 1
			}
		\end{mathpar}
		Notice that in the conclusion of the rule, in both examples, we remove $0$ from the set of marker orders, because the order-$0$ marker is provided by $x$.
	\end{example}

	The next two rules use a predicate $\Comp_m$, saying how flags and markers from premisses contribute to the conclusion.
	It takes ``as input'' pairs $(F_i,c_i)$ for $i\in I$; each of them consists of the set of flag orders $F_i$ and of the flag counter $c_i$ from some premiss.
	Moreover, the predicate takes a set of marker orders $M$ from the current type judgment (it contains orders of markers used in the derivation, including those provided by free variables).
	The goal is to compute the set of flag orders $F$ and the flag counter $c$ that should be placed in the current type judgment.
	First, for each $n\in\{1,\dots,m\}$ consecutively, we decide whether a flag of order $n$ should be placed on the current type judgment.
	We follow here the rules mentioned in the intuitive description.
	Namely, we place a flag of order $n$ if we are on the path leading to the marker of order $n-1$ (i.e., if $n-1\in M$), and simultaneously we receive an information about a flag of order $n-1$.
	By receiving this information we mean that either a flag of order $n-1$ was placed on the current type judgment, or $n-1$ belongs to some set $F_i$.
	Actually, we place multiple flags of order $n$: one per each flag of order $n-1$ placed on the current type judgment, and one per each set $F_i$ containing $n-1$.
	Then, we compute $F$ and $c$.
	In $c$ we store the number of flags of the maximal order $m$: we sum all the numbers $c_i$, and we add the number of order-$m$ flags placed on the current type judgment.
	In $F$ we keep elements of all $F_i$, and we add the orders $n$ of flags that were placed on the current type judgment.
	We, however, remove from $F$ all elements of $M$.
	This is because every flag of some order $n-1$ should result in creating at most one flag of order $n$, in the closest ancestor that lies on the path leading to the marker of order $n-1$.
	If we have created an order-$n$ flag on the current type judgment, i.e., if $n-1\in M$, we do not want to do this again in the parent.

	Below we give a formal definition, in which $f_n'$ contains the number of order-$n$ flags placed on the current type judgment,
	while $f_n$ additionally counts the number of premisses for which $n\in F_i$.
	We say that $\Comp_m(M\palkaC((F_i,c_i))_{i\in I})=(F,c)$ when
	\begin{align*}
		&F=\{n\in\{0,\dots,m-1\}\mid f_n>0\land n\not\in M\}\,,&&c=f'_m+\sum_{i\in I}c_i\,,\qquad\mbox{where, for $n\in\{0,\dots,m\}$,}\\
		&f_n=f_n'+\sum_{i\in I}|F_i\cap\{n\}|,&&f_n'=\left\{\begin{array}{ll}f_{n-1}&\mbox{if }n-1\in M,\\0&\mbox{otherwise.}\end{array}\right.
	\end{align*}

	We now present a rule for constants other than $\br$:
	\begin{mathpar}
		\inferrule*[right=(Con)]{
			\Gamma_i\vdash P_i:(m,F_i,M_i,o)\triangleright c_i\mbox{ for each }i\in\{1,\dots,r\}
			\\
			M=M'\uplus M_1\uplus\dots\uplus M_r
			\\
			(m=0)\Rightarrow(F'=\emptyset\land c'=1)
			\\ 
			(m>0)\Rightarrow(F'=\{0\}\land c'=0)
			\\
			(r>0)\Rightarrow(M'=\emptyset)
			\\
			a\neq\br
			\\
			\Split(\Gamma\palka\Gamma_1,\dots,\Gamma_r)
			\\ 
			\Comp_m(M\palkaC(F',c'),(F_1,c_1),\dots,(F_r,c_r))=(F,c)
			}
			{\Gamma\vdash a\,P_1\,\dots\,P_r:(m,F,M,o)\triangleright c}
	\end{mathpar}

	Here, the conditions in the second line say that in a node using the \ConR rule we always place a flag of order $0$ (via $F'$ or via $c'$, depending on $m$), 
	and that if the node is a leaf (i.e., $r=0$), then we are allowed to place markers of arbitrary order (via $M'$).
	Then to the $\Comp_m$ predicate, beside of pairs $(F_i,c_i)$ coming from premisses, we also pass the information $(F',c')$ about the order-$0$ flag placed in the current node;
	this predicate decides whether we should place also some flags of positive orders.
	Let us emphasize that in this rule (and similarly in the next rule) we have a disjoint union $M'\uplus M_1\uplus\dots\uplus M_r$,
	which ensures that a marker of any order may be placed only in one node of a derivation.
	
	\begin{example}\label{ex:con}
		The \ConR rule may be instantiated in the following way:
		\begin{mathpar}
			\inferrule*[Right=(Con)]{
				\Gammae[x\mapsto\{\hat\rho_1\}]\vdash x:(2,\emptyset,\{0\},o)\triangleright 0
			}{
				\Gammae[x\mapsto\{\hat\rho_1\}]\vdash a\,x:(2,\{1\},\{0\},o)\triangleright 0
			}
			\and
			\inferrule*[Right=(Con)]{
				\Gammae[x\mapsto\{\hat\rho_1\}]\vdash x:(2,\emptyset,\{0,1\},o)\triangleright 0
			}{
				\Gammae[x\mapsto\{\hat\rho_1\}]\vdash a\,x:(2,\emptyset,\{0,1\},o)\triangleright 1
			}
		\end{mathpar}
		In the left example, flags of order $0$ and $1$ are placed in the conclusion of the rule 
		(a flag of order $0$ is created because we are in a constant; since the marker of order $0$ is visible, we do not put $0$ into the set of flag orders, but instead we create a flag of order $1$).
		In the right example, a marker of order $1$ is visible, which causes that this time flags of order $0$, $1$, and $2$ are placed in the conclusion of the \ConR rule
		(again, we do not put $0$ nor $1$ into the set of flag orders, because of $0$ and $1$ in the set of marker orders).
	\end{example}

	The next rule describes application:
	\begin{mathpar}
		\inferrule*[right=(\!@\!)]{
			\Gamma'\vdash P:(m,F',M',\{(\ord(P),F_i\restr_{<\ord(P)},M_i\restr_{<\ord(P)},\tau_i)\mid i\in I\}\arr\tau)\triangleright c'
			\\
			\Gamma_i\vdash Q:(m,F_i,M_i,\tau_i)\triangleright c_i\mbox{ for each }i\in I
			\\
 			M=M'\uplus\biguplus{}_{i\in I}M_i
 			\\
 			\ord(P)\leq m
 			\\
 			\Split(\Gamma\palka\Gamma',(\Gamma_i)_{i\in I})
 			\\
 			\Comp_m(M\palkaC(F',c'),((F_i\restr_{\geq\ord(P)},c_i))_{i\in I})=(F,c)
	       		}
		        {\Gamma\vdash P\,Q:(m,F,M,\tau)\triangleright c}
	\end{mathpar}

	In this rule, it is allowed (but in fact useless) that for two different $i\in I$ the full types $(m,F_i,M_i,\tau_i)$ are equal.
	It is also allowed that $I=\emptyset$, in which case no type needs to be derived for $Q$.
	Observe how flags and markers coming from premisses concerning $Q$ are propagated: only flags and markers of order $n<\ord(P)$ are visible to $P$, while only flags of order $n\geq\ord(P)$ are passed to the $\Comp_m$ predicate.
	This can be justified if we recall the intuitions staying behind the type system (see page \pageref{para:intuitions}).
	Indeed, while considering flags and markers of order $n$, we should imagine the $\lambda$-term obtained from the current $\lambda$-term by performing all $\beta$-reductions of all orders greater than $n$;
	the distribution of flags and markers of order $n$ in the current $\lambda$-term actually simulates their distribution in this imaginary $\lambda$-term.
	Thus, if $n<\ord(P)$, then our application will disappear in this imaginary $\lambda$-term, and $Q$ will be already substituted somewhere in $P$; 
	for this reason we need to pass the information about flags and markers of order $n$ from $Q$ to $P$.
	Conversely, if $n\geq\ord(P)$, then in the imaginary $\lambda$-term the considered application will be still present, 
	and in consequence the subterm corresponding to $P$ will not see flags and markers of order $n$ placed in the subterm corresponding to $Q$.

	\begin{example}\label{ex:app}
		Denote by $\hat\tau_\mathsf{f}$ and $\hat\tau_\mathsf{m}$ the types derived in Example \ref{ex:lambda}:
		\begin{align*}
			&\hat\tau_\mathsf{f}=(2,\{1\},\emptyset,\{\hat\rho_1\}\arr o)\,,&
			&\mbox{and}&
			&\hat\tau_\mathsf{m}=(2,\emptyset,\{1\},\{\hat\rho_1\}\arr o)\,.
		\end{align*}
		Then, using the \AppR rule, we can derive (where $e$ is a symbol of rank $0$, and $f$ a variable):
		\begin{mathpar}
			\inferrule*[Right=(\!@\!)]{
				\inferrule*[right=(Var)]{ }{
					\Gammae[f\mapsto\{\hat\tau_\mathsf{m}\}]\vdash f:\hat\tau_\mathsf{m}\triangleright 0
				}
				\and
				\inferrule*[Right=(Con)]{ }{
					\Gammae\vdash e:(2,\{1\},\{0\},o)\triangleright 0
				}
			}{
				\Gammae[f\mapsto\{\hat\tau_\mathsf{f},\hat\tau_\mathsf{m}\}]\vdash f\,e:(2,\emptyset,\{0,1\},o)\triangleright 1
			}
		\end{mathpar}
		Recall that $\hat\rho_1=(1,\emptyset,\{0\},o)$.
		In the conclusion of the \AppR rule the information about a flag of order $1$ (from the second premiss) meets the information about the marker of order $1$ (from the first premiss), 
		and thus a flag of order $2$ is placed, which increases the flag counter.
		Notice that we have discarded the full type $\hat\tau_\mathsf{f}$ assigned to $f$ in the type environment;
		this is allowed because $\hat\tau_\mathsf{f}$ provides no markers (equally well $\hat\tau_\mathsf{f}$ could be assigned to $f$ also in one or two of the premisses, and discarded there).
		On the other hand, the full type $\hat\tau_\mathsf{m}$ provides markers, so it cannot be discarded nor duplicated (in particular, we could not pass it to the conclusion of the \ConR rule).
	\end{example}
	
	The key property of the type system is described by the following theorem.
	
	\begin{theorem}\label{thm:types-ok}
		Let $P$ be a closed $\lambda$-term of sort $o$ and complexity $m$.
		Then $\Ll(P)$ is infinite if and only if for arbitrarily large $c$ we can derive $\Gammae\vdash P:\hat\rho_m\triangleright c$, where $\hat\rho_m=(m,\emptyset,\{0,\dots,m-1\},o)$.
	\end{theorem}

	The left-to-right implication of Theorem \ref{thm:types-ok} (completeness of the type system) is shown in Section \ref{sec:compl}, while the opposite implication (soundness of the type system) in Section \ref{sec:sound}.
	In Section \ref{sec:effective} we discuss how Theorem \ref{thm:main} follows from Theorem \ref{thm:types-ok}.
	Before all that, we give a few more examples of derivations, illustrating the type system and Theorem \ref{thm:types-ok}.
		
	\begin{example}\label{ex:large-1}
		In this example we analyze the $\lambda$-term $P_1=R\,(\lambda x.a\,x)$, where $R$ is defined by coinduction as $R=(\lambda f.\br\,(f\,e)\,(R\,(\lambda x.f\,(f\,x))))$.
		As previously, $a$ and $e$ are symbols of rank $1$ and $0$, respectively.
		In $\Ll(P_1)$ there are trees that consist of a branch of $a$ symbols ended with an $e$ symbol, but only those where the number of $a$ symbols is $2^k$ for some $k\in\Nat$.
		Notice that the complexity of $P_1$ is $2$.
		
		Continuing Example \ref{ex:app}, we derive the full type $\hat\sigma_R=(2,\emptyset,\{0\},\{\hat\tau_\mathsf{f},\hat\tau_\mathsf{m}\}\arr o)$ for $R$:
		\begin{mathpar}
			\inferrule*[Right=($\lambda$)]{
				\inferrule*[Right=(Br)]{
					\Gammae[f\mapsto\{\hat\tau_\mathsf{f},\hat\tau_\mathsf{m}\}]\vdash f\,e:(2,\emptyset,\{0,1\},o)\triangleright 1
				}{
					\Gammae[f\mapsto\{\hat\tau_\mathsf{f},\hat\tau_\mathsf{m}\}]\vdash \br\,(f\,e)\,(R\,(\lambda x.f\,(f\,x))):(2,\emptyset,\{0,1\},o)\triangleright 1
				}
			}{
				\Gammae\vdash R:\hat\sigma_R\triangleright 1
			}
		\end{mathpar}
	
		Next, we derive the same full type for $R$, but using the second argument of the $\br$ symbol; this results in greater values of the flag counter.
		We start by deriving the full type $\hat\tau_\mathsf{f}$ for the subterm $\lambda x.f\,(f\,x)$:
		\begin{mathpar}
			\inferrule*[right=($\lambda$)]{
				\inferrule*[Right=(\!@\!)]{
					\inferrule*{ }{
						\Gammae[f\mapsto\{\hat\tau_\mathsf{f}\}]\vdash f:\hat\tau_\mathsf{f}\triangleright 0
					}
					\and
					\inferrule*[Right=(\!@\!)]{
						\inferrule*{ }{
							\Gammae[f\mapsto\{\hat\tau_\mathsf{f}\}]\vdash f:\hat\tau_\mathsf{f}\triangleright 0
						}
						\and
						\inferrule*{ }{
							\Gammae[x\mapsto\{\hat\rho_1\}]\vdash x:(2,\emptyset,\{0\},o)\triangleright 0
						}
					}{
						\Gammae[f\mapsto\{\hat\tau_\mathsf{f}\},x\mapsto\{\hat\rho_1\}]\vdash f\,x:(2,\{1\},\{0\},o)\triangleright 0
					}
				}{
					\Gammae[f\mapsto\{\hat\tau_\mathsf{f}\},x\mapsto\{\hat\rho_1\}]\vdash f\,(f\,x):(2,\{1\},\{0\},o)\triangleright 0
				}
			}{
				\Gammae[f\mapsto\{\hat\tau_\mathsf{f}\}]\vdash\lambda x.f\,(f\,x):\hat\tau_\mathsf{f}\triangleright 0
			}
		\end{mathpar}
		In the above derivation there are no flags nor markers.
		Next, we derive $\hat\tau_\mathsf{m}$ for the same subterm:
		\begin{mathpar}
			\inferrule*[right=($\lambda$)]{
				\inferrule*[Right=(\!@\!)]{
					\inferrule*{ }{
						\Gammae[f\mapsto\{\hat\tau_\mathsf{f}\}]\vdash f:\hat\tau_\mathsf{f}\triangleright 0
					}
					\and
					\inferrule*[Right=(\!@\!)]{
						\inferrule*{ }{
							\Gammae[f\mapsto\{\hat\tau_\mathsf{m}\}]\vdash f:\hat\tau_\mathsf{m}\triangleright 0
						}
						\and
						\inferrule*{ }{
							\Gammae[x\mapsto\{\hat\rho_1\}]\vdash x:(2,\emptyset,\{0\},o)\triangleright 0
						}
					}{
						\Gammae[f\mapsto\{\hat\tau_\mathsf{m}\},x\mapsto\{\hat\rho_1\}]\vdash f\,x:(2,\emptyset,\{0,1\},o)\triangleright 0
					}
				}{
					\Gammae[f\mapsto\{\hat\tau_\mathsf{f},\hat\tau_\mathsf{m}\},x\mapsto\{\hat\rho_1\}]\vdash f\,(f\,x):(2,\emptyset,\{0,1\},o)\triangleright 1
				}
			}{
				\Gammae[f\mapsto\{\hat\tau_\mathsf{f},\hat\tau_\mathsf{m}\}]\vdash\lambda x.f\,(f\,x):\hat\tau_\mathsf{m}\triangleright 1
			}
		\end{mathpar}
		Below the lower \AppR rule the information about a flag of order $1$ meets the information about the marker of order $1$, and thus a flag of order $2$ is placed, which increases the flag counter.
		We continue with the $\lambda$-term $R$:
		\begin{mathpar}
			\inferrule*[right=($\lambda$)]{
				\inferrule*[Right=(Br)]{
					\inferrule*[Right=(\!@\!)]{
						\Gammae\vdash R:\hat\sigma_R\triangleright c
						\and
						\Gammae[f\mapsto\{\hat\tau_\mathsf{f}\}]\vdash\lambda x.f\,(f\,x):\hat\tau_\mathsf{f}\triangleright 0
						\and
						\Gammae[f\mapsto\{\hat\tau_\mathsf{f},\hat\tau_\mathsf{m}\}]\vdash\lambda x.f\,(f\,x):\hat\tau_\mathsf{m}\triangleright 1
					}{
						\Gammae[f\mapsto\{\hat\tau_\mathsf{f},\hat\tau_\mathsf{m}\}]\vdash R\,(\lambda x.f\,(f\,x)):(2,\emptyset,\{0,1\},o)\triangleright c+1
					}
				}{
					\Gammae[f\mapsto\{\hat\tau_\mathsf{f},\hat\tau_\mathsf{m}\}]\vdash \br\,(f\,e)\,(R\,(\lambda x.f\,(f\,x))):(2,\emptyset,\{0,1\},o)\triangleright c+1
				}
			}{
				\Gammae\vdash R:\hat\sigma_R\triangleright c+1
			}
		\end{mathpar}
		In this fragment of a derivation no flag nor marker is placed.
		In particular, there is no order-$2$ flag in conclusion of the \AppR rule, although its second premiss provides a flag of order $1$ while the third premiss provides the marker of order $1$.
		We recall from the definition of the \AppR rule that the information about flags and markers coming from the arguments is divided into two parts.
		Numbers smaller than the order of the operator ($\ord(R)=2$ in our case) are passed to the operator, while only greater numbers ($\geq 2$ in our case) contribute in creating new flags via the $\Comp$ predicate.
		
		By composing the above fragments of a derivation, we can derive $\Gammae\vdash R:\hat\sigma_R\triangleright c$ for every $c\geq 1$.
		Recall that in Examples \ref{ex:var}-\ref{ex:con} we have derived $\Gammae\vdash\lambda x.a\,x:\hat\tau_\mathsf{f}\triangleright 0$ and $\Gammae\vdash\lambda x.a\,x:\hat\tau_\mathsf{m}\triangleright 1$.
		Together with the above, this allows to derive for $P_1$ the full type $\hat\rho_2=(2,\emptyset,\{0,1\},o)$ (appearing in Theorem \ref{thm:types-ok}):
		\begin{mathpar}
			\inferrule*[Right=(\!@\!)]{
				\Gammae\vdash R:\hat\sigma_R\triangleright c
				\and
	                        \Gammae\vdash\lambda x.a\,x:\hat\tau_\mathsf{f}\triangleright 0
				\and
	                        \Gammae\vdash\lambda x.a\,x:\hat\tau_\mathsf{m}\triangleright 1
			}{
				\Gammae\vdash P_1:\hat\rho_2\triangleright c+1
			}
		\end{mathpar}
		We can notice a correspondence between a derivation with flag counter $c+1$ and a tree in $\Ll(P)$ of size $2^{c-1}+1$.
		We remark that in every of these derivations only three flags of order $0$ and only three flags of order $1$ are present, in the three nodes using the \ConR rule.
	\end{example}
	
	\begin{example}
		Consider a similar $\lambda$-term $P_2=R\,(\lambda x.b\,x\,x)$, where $R$ is as previously, and $b$ is a symbol of rank $2$.
		In $\Ll(P_2)$ we have, for every $k\in\Nat$, a full binary tree in which every branch consist of $2^k$ symbols $b$ and ends with an $e$ symbol.
	
		This time for the subterm $\lambda x.b\,x\,x$ we need to derive three full types:
		\begin{align*}
			&\hat\tau_0'=(2,\{0\},\emptyset,\{(1,\{0\},\emptyset,o)\}\arr o)\,,\\
			&\hat\tau_\mathsf{f}'=(2,\{1\},\emptyset,\{(1,\{0\},\emptyset,o),\hat\rho_1\}\arr o)\,,\qquad\mbox{and}\\
			&\hat\tau_\mathsf{m}'=(2,\emptyset,\{1\},\{(1,\{0\},\emptyset,o),\hat\rho_1\}\arr o)\,.
		\end{align*}
		The last one is derived with flag counter $1$.
		Notice that $\hat\tau_\mathsf{f}'$ and $\hat\tau_\mathsf{m}'$ need now two full types for the argument $x$; the new one $(1,\{0\},\emptyset,o)$ describes the subtree that is not on the path to the order-$0$ marker.
		We also have a new full type $\hat\tau_0'$ that describes the use of $\lambda x.b\,x\,x$ outside of the path to the order-$0$ marker.
			
		Then, similarly as in the previous example, for every $c\geq 1$ we can derive $\Gammae\vdash R:\hat\sigma_R'\triangleright c$, 
		where $\hat\sigma_R'=(2,\emptyset,\{0\},\{\hat\tau_0',\hat\tau_\mathsf{f}',\hat\tau_\mathsf{m}'\}\arr o)$.
		Again, this allows to derive $\Gammae\vdash P_2:\hat\rho_2\triangleright c+1$.
		This time a derivation with flag counter $c+1$ corresponds to a tree in $\Ll(P)$ of size $2^h-1$ with $h=2^{c-1}+1$.
	\end{example}
	
	\begin{example}
		Next, consider the $\lambda$-term $P_3=R\,(\lambda x.\,x)$.
		The only tree in $\Ll(P_3)$ consists of a single $e$ node.
		Let us see how the derivation from Example \ref{ex:large-1} has to be modified.
		The full type $\hat\tau_\mathsf{m}$ can still be derived for $\lambda x.\,x$ (although with flag counter $0$ now), 
		but instead of $\hat\tau_\mathsf{f}$ we have to use $\hat\tau_\mathsf{f}''=(2,\emptyset,\emptyset,\{\hat\rho_1\}\arr o)$ that provides no flag of order $1$:
		\begin{mathpar}
			\inferrule*[Right=($\lambda$)]{
				\inferrule*[Right=(Var)]{ }{
					\Gammae[x\mapsto\{\hat\rho_1\}]\vdash x:(2,\emptyset,\{0\},o)\triangleright 0
				}
			}{
				\Gammae\vdash\lambda x.x:\hat\tau_\mathsf{f}''\triangleright 0
			}
			\and
			\inferrule*[Right=($\lambda$)]{
				\inferrule*[Right=(Var)]{ }{
					\Gammae[x\mapsto\{\hat\rho_1\}]\vdash x:(2,\emptyset,\{0,1\},o)\triangleright 0
				}
			}{
				\Gammae\vdash\lambda x.x:\hat\tau_\mathsf{m}\triangleright 0
			}
		\end{mathpar}
		
		Next, for $R$ we want to derive the full type $\hat\sigma_R''=(2,\emptyset,\{0\},\{\hat\tau_\mathsf{f}'',\hat\tau_\mathsf{m}\}\arr o)$.
		We can easily adopt every of the previous derivations for $\Gammae\vdash R:\hat\sigma_R\triangleright c$: we basically replace every $\hat\tau_\mathsf{f}$ by $\hat\tau_\mathsf{f}''$.
		The key point is that while deriving the full type $\hat\tau_\mathsf{m}$ for the subterm $\lambda x.f\,(f\,x)$, previously in the lower \AppR rule we have received information about an order-$1$ flag,
		and thus we have created an order-$2$ flag and increased the flag counter;
		this time there is no information about an order-$1$ flag, and thus we do not create an order-$2$ flag and do not increase the flag counter.
		In consequence, even if this part of the derivation is repeated arbitrarily many times, the value of the flag counter of the whole derivation remains $1$.
	\end{example}
	
	\begin{example}
		Finally, consider the $\lambda$-term $P_4=(\lambda g.P_3)\,(\lambda x.a\,(a\,(\dots\,(a\,x)\dots))$, which $\beta$-reduces to $P_3$.
		Notice that we can create the following derivation:
		\begin{mathpar}
			\inferrule*[Right=($\lambda$)]{
				\inferrule*[Right=(Con)]{
					\inferrule*[Right=(Con)]{ 
						\inferrule*[Right=(Con)]{ 
							\inferrule*[Right=(Var)]{ }{
								\Gammae[x\mapsto\{\hat\rho_1\}]\vdash x:(2,\emptyset,\{0\},o)\triangleright 0
							}
						}{
							\Gammae[x\mapsto\{\hat\rho_1\}]\vdash a\,x:(2,\{1\},\{0\},o)\triangleright 0
						}
					}{
						\vdots
					}
				}{
					\Gammae[x\mapsto\{\hat\rho_1\}]\vdash a\,(a\,(\dots\,(a\,x)\dots)):(2,\{1\},\{0\},o)\triangleright 0
				}
			}{
				\Gammae\vdash\lambda x.a\,(a\,(\dots\,(a\,x)\dots)):\hat\tau_\mathsf{f}\triangleright 0
			}
		\end{mathpar}
		Every \ConR rule used in this derivation places in its conclusion an order-$0$ flag and an order-$1$ flag.
		This derivation can be used as a part of a derivation for $P_4$:
		\begin{mathpar}
			\inferrule*[Right=(\!@\!)]{
				\inferrule*[right=($\lambda$)]{
					\Gammae[g\mapsto\{\hat\tau_\mathsf{f}\}]\vdash P_3:\hat\rho_2\triangleright 1
				}{
					\Gammae\vdash\lambda g.P_3:(2,\emptyset,\{0,1\},\{\hat\tau_\mathsf{f}\}\arr o)\triangleright 1
				}
				\and
				\Gammae\vdash\lambda x.a\,(a\,(\dots\,(a\,x)\dots)):\hat\tau_\mathsf{f}\triangleright 0
			}{
				\Gammae\vdash P_4:\hat\rho_2\triangleright 1
			}
		\end{mathpar}
		Because $\hat\tau_\mathsf{f}$ provides no markers, it can be removed from the type environment and thus for $P_3$ we can use the derivation from the previous example.
		We thus obtain a derivation for $P_4$ in which there are many order-$0$ and order-$1$ flags (but only one flag of order $2$).
		This shows that in the flag counter we indeed need to count only the number of flags of the maximal order (not, say, the total number of flags of all orders).
	\end{example}

\section{Completeness}\label{sec:compl}

	The proof of the left-to-right implication of Theorem \ref{thm:types-ok} is divided into the following three lemmata.
	Recall that a $\beta$-reduction $P\bred Q$ is of order $n$ if it concerns a redex $(\lambda x.R)\,S$ such that $\ord(\lambda x.R)=n$.
	The number of nodes of a tree $t$ is denoted $|t|$.
	As in Theorem \ref{thm:types-ok}, we denote $\hat\rho_m=(m,\emptyset,\{0,\dots,m-1\},o)$.
	
	\begin{lemma}\label{lem:base}
		Let $P$ be a closed $\lambda$-term of sort $o$ and complexity $m$, and let $t\in\Ll(P)$.
		Then there exist $\lambda$-terms $Q_m,Q_{m-1},\dots,Q_0$ such that $P=Q_m$, and for every $k\in\{1,\dots,m\}$ the term $Q_{k-1}$ can be reached from $Q_k$ using only $\beta$-reductions of order $k$, 
		and we can derive $\Gammae\vdash Q_0:\hat\rho_0\triangleright|t|$.
	\end{lemma}
	
	\begin{lemma}\label{lem:increase-m}
		Suppose that we can derive $\Gammae\vdash P:\hat\rho_m\triangleright c$.
		Then we can also derive $\Gammae\vdash P:\hat\rho_{m+1}\triangleright c'$ for some $c'\geq\log_2 c$.
	\end{lemma}
	
	\begin{lemma}\label{lem:c-step}
		Suppose that $P\bred Q$ is a $\beta$-reduction of order $m$, and we can derive $\Gamma\vdash Q:\hat\tau\triangleright c$ with $\ord(\hat\tau)=m$.
		Then we can also derive $\Gamma\vdash P:\hat\tau\triangleright c$.
	\end{lemma}
	
	Now the left-to-right implication of Theorem \ref{thm:types-ok} easily follows.
	Indeed, take a closed $\lambda$-term $P$ of sort $o$ and complexity $m$ such that $\Ll(P)$ is infinite, and take any $c\in\Nat$.
	By $\log^k_2$ we denote the $k$-fold application of the logarithm: $\log^0_2 x=x$ and $\log^{k+1}_2 x=\log_2(\log_2^k x)$.
	Since $\Ll(P)$ is infinite, it contains a tree $t$ so big that $\log_2^m|t|\geq c$.
	We apply Lemma \ref{lem:base} to this tree, obtaining $\lambda$-terms $Q_m,Q_{m-1},\dots,Q_0$ and a derivation of $\Gammae\vdash Q_0:\hat\rho_0\triangleright|t|$.
	Then repeatedly for every $k\in\{1,\dots,m\}$ we apply Lemma \ref{lem:increase-m}, obtaining a derivation of $\Gammae\vdash Q_{k-1}:\hat\rho_k\triangleright c_k$ for some $c_k\geq\log^k_2|t|$,
	and Lemma \ref{lem:c-step} for every $\beta$-reduction (of order $k$) between $Q_k$ and $Q_{k-1}$, obtaining a derivation of $\Gammae\vdash Q_k:\hat\rho_k\triangleright c_k$.
	We end with a derivation of $\Gammae\vdash P:\hat\rho_m\triangleright c_m$, where $c_m\geq\log^m_2|t|\geq c$, as needed.
	In the remaining part of this section we prove the three lemmata.
	
	\begin{proof}[Proof of Lemma \ref{lem:base} (sketch)]
		Recall that $t\in\Ll(P)$ is a finite tree, thus it can be found in some finite prefix of the B\"ohm tree of $P$.
		By definition, this prefix will be already expanded after performing some finite number of $\beta$-reductions from $P$.
		We need to observe that these $\beta$-reductions can be rearranged, so that those of higher order are performed first.
		\label{page:rearrange-beta}
		The key point is to observe that when we perform a $\beta$-reduction of some order $k$, then no new $\beta$-redexes of higher order appear in the term.
		Indeed, suppose that $(\lambda x.R)\,S$ is changed into $R[S/x]$ somewhere in a term, where $\ord(\lambda x.R)=k$.
		One new redex that may appear is when $R$ starts with a $\lambda$, and to the whole $R[S/x]$ some argument is applied; this redex is of order $\ord(R)\leq k$.
		Some other redexes may appear when $S$ starts with a $\lambda$, and is substituted for such appearance of $x$ to which some argument is applied; but this redex is of order $\ord(S)<k$.
		
		We can thus find a sequence of $\beta$-reductions in which $\beta$-reductions are arranged according to their order, that leads from $P$ to some $Q_0$ such that $t$ can be found in the prefix of $Q_0$ that is already expanded to a tree.
		It is now a routine to use the rules of our type system and derive $\Gammae\vdash Q_0:\hat\rho_0\triangleright|t|$:
		in every $\br$-labeled node we choose the subtree in which $t$ continues, and this effects in counting the number of nodes of $t$ in the flag counter.
	\end{proof}
	
	\begin{proof}[Proof of Lemma \ref{lem:increase-m}]
		Consider some derivation of $\Gammae\vdash P:\hat\rho_m\triangleright c$.
		In this derivation we choose a leaf in which we will put the order-$m$ marker, as follows.
		Starting from the root of the derivation, we repeatedly go to this premiss in which the flag counter is the greatest (arbitrarily in the case of a tie).
		In every node that is not on the path to the selected leaf, we replace the current type judgment $\Gamma\vdash Q:(m,F,M,\tau)\triangleright d$ by $\Gamma\vdash Q:(m+1,F',M,\tau)\triangleright 0$,
		where $F'=F\cup\{m\}$ if $d>0$, and $F'=F$ otherwise.
		In the selected leaf and all its ancestors, we change the order from $m$ to $m+1$, we add $m$ to the set of marker orders, and we recalculate the flag counter.
		
		Let us see how such transformation changes the flag counter on the path to the selected leaf.
		We will prove (by induction) that the previous value $d$ and the new value $d'$ of the flag counter in every node on this path satisfy $d'\geq\log_2 d$.
		In the selected leaf itself, the flag counter (being either $0$ or $1$) remains unchanged; we have $d'=d\geq\log_2 d$.
		Next, consider any proper ancestor of the selected node.
		Let $k$ be the number of those of its children in which the flag counter was positive, plus the number of order-$m$ flags placed in the considered node itself.
		Let also $d_{\max}$ and $d_{\max}'$ be the previous value and the new value of the flag counter in this child that is in the direction of the selected leaf.
		By construction, the flag counter in this child was maximal, which implies $k\cdot d_{\max}\geq d$, while by the induction assumption $d'_{\max}\geq\log_2 d_{\max}$.
		To $d'$ we take the flag counter only from the special child, while for other children with positive flag counter we add $1$, i.e., $d'=k-1+d'_{\max}$.
		Altogether we obtain $d'=k-1+d'_{\max}\geq k-1+\log_2d_{\max}\geq\log_2(k\cdot d_{\max})\geq\log_2 d$, as required.
	\end{proof}
	
	\begin{proof}[Proof of Lemma \ref{lem:c-step}]
		We consider the base case when $P=(\lambda x.R)\,S$ and $Q=R[S/x]$; the general situation (redex being deeper in $P$) is easily reduced to this one.
		In the derivation of $\Gamma\vdash Q:\hat\tau\triangleright c$ we identify the set $I$ of places (nodes) where we derive a type for $S$ substituted for $x$.
		For $i\in I$, let $\Sigma_i\vdash S:\hat\sigma_i\triangleright d_i$ be the type judgment in $i$.
		We change the nodes in $I$ into leaves, where we instead derive $\Gammae[x\mapsto\{\hat\sigma_i\}]\vdash x:\hat\sigma_i\triangleright 0$.
		It should be clear that we can repair the rest of the derivation, by changing type environments, replacing $S$ by $x$ in $\lambda$-terms, and decreasing flag counters.
		In this way we obtain derivations of $\Sigma_i\vdash S:\hat\sigma_i\triangleright d_i$ for every $i\in I$, and a derivation of $\Sigma'\vdash R:\hat\tau\triangleright d$,
		where $\Sigma'=\Sigma[x\mapsto\{\hat\sigma_i\mid i\in I\}]$ with $\Sigma(x)=\emptyset$, 
		and $\Split(\Gamma\palka\Sigma,(\Sigma_i)_{i\in I})$, and $c=d+\Sigma_{i\in I}d_i$.
		To the latter type judgment we apply the $\LamR$ rule, and then we merge it with the type judgments for $S$ using the \AppR rule, which results in a derivation for $\Gamma\vdash P:\hat\tau\triangleright c$.
		We remark that different $i\in I$ may give identical type judgments for $S$ (as long as the set of markers in $\hat\sigma_i$ is empty); this is not a problem.
		The \AppR rule requires that $\ord(\hat\sigma_i)=\ord(\lambda x.R)$; we have that $\ord(\hat\sigma_i)=\ord(\hat\tau)$, and $\ord(\hat\tau)=m=\ord(\lambda x.R)$ by assumption.
	\end{proof}

\section{Soundness}\label{sec:sound}

	In this section we sketch the proof of the right-to-left implication of Theorem \ref{thm:types-ok}.
	We, basically, need to reverse the proof from the previous section.
	The following new fact is now needed.

	\begin{lemma}\label{lem:zero-when-no-marker}
		If we can derive $\Gamma\vdash P:(m,F,M,\tau)\triangleright c$ with $m-1\not\in M$ and $\ord(P)\leq m-1$, then $c=0$.
	\end{lemma}

	A simple inductive proof is based on the following idea: 
	flags of order $m$ are created only when a marker of order $m-1$ is visible;
	the derivation itself (together with free variables) does not provide it ($m-1\not\in M$), and the arguments, i.e.~sets $T_1,\dots,T_k$ in $\tau=T_1\arr\dots\arr T_k\arr o$, 
	may provide only markers of order at most $\ord(P)-1\leq m-2$ (see the definition of a type), thus no flags of order $m$ can be created.
	
	We say that a $\lambda$-term of the form $P\,Q$ is an application \emph{of order $n$} when $\ord(P)=n$, and that an \AppR rule is \emph{of order $n$} if it derives a type for an application of order $n$.
	We can successively remove applications of the maximal order from a type derivation.
	
	\begin{lemma}\label{lem:s-step}
		Suppose that $\Gammae\vdash P:\hat\rho_m\triangleright c$ for $m>0$ is derived by a derivation $D$ in which the \AppR rule of order $m$ is used $n$ times.
		Then there exists $Q$ such that $P\bred Q$ and $\Gammae\vdash Q:\hat\rho_m\triangleright c$ can be derived by a derivation $D'$ in which the \AppR rule of order $m$ is used less than $n$ times.
	\end{lemma}
	
	Recall from the definition of the type system that the \AppR rule of orders higher than $m$ cannot be used while deriving a full type of order $m$.
	Thus in $D$ we have type judgments only for subterms of $P$ of order at most $m$ (although $P$ may also have subterms of higher orders), 
	and in type environments we only have variables of order at most $m-1$.
	In order to prove Lemma \ref{lem:s-step} we choose in $P$ a subterm $R\,S$ with $\ord(R)=m$ such that there is a type judgment for $R\,S$ in some nodes of $D$ (at least one), 
	but no descendants of those nodes use the \AppR rule of order $m$.
	Since $R$ is of order $m$, it cannot be an application (then we would choose it instead of $R\,S$) nor a variable; thus $R=\lambda x.R'$.
	We obtain $Q$ by reducing the redex $(\lambda x.R')\,S$; the derivation $D'$ is obtained by performing a surgery on $D$ similar to that in the proof of Lemma \ref{lem:c-step} (but in the opposite direction).
	Notice that every full type $(m,F,M,\tau)$ (derived for $S$) with nonempty $M$ is used for exactly one appearance of $x$ in the derivation for $R'$;
	full types with empty $M$ may be used many times, or not used at all, but thanks to Lemma \ref{lem:zero-when-no-marker} duplicating or removing the corresponding derivations for $S$ does not change the flag counter.
	In the derivations for $R'[S/x]$ no \AppR rule of order $m$ may appear, and the application $R\,S$ disappears, so the total number of \AppR rules of order $m$ decreases.
	
	When all \AppR rules of order $m$ are eliminated, we can decrease $m$.

	\begin{lemma}\label{lem:s-zero}
		Suppose that $\Gammae\vdash P:\hat\rho_m\triangleright c$ for $m>0$ is derived by a derivation $D$ in which the \AppR rule of order $m$ is not used.
		Then we can also derive $\Gammae\vdash P:\hat\rho_{m-1}\triangleright c'$ for some $c'\geq c$.
	\end{lemma}
	
	The proof is easy; we simply decrease the order $m$ of all derived full types by $1$, and we ignore flags of order $m$ and markers of order $m-1$.
	To obtain the inequality $c'\geq c$ we observe that when no \AppR rule of order $m$ is used, the information about flags of order $m-1$ goes only from descendants to ancestors,
	and thus every flag of order $m$ is created because of a different flag of order $m-1$.

	By repeatedly applying the two above lemmata, out of a derivation of $\Gammae\vdash P:\hat\rho_m\triangleright c$ we obtain a derivation of $\Gammae\vdash Q:\hat\rho_0\triangleright c'$, where $P\bred^*Q$ and $c'\geq c$.
	Since $\hat\rho_0$ is of order $0$, using the latter derivation it is easy to find in the already expanded part of $Q$ (and thus in $\Ll(Q)=\Ll(P)$) a tree $t$ such that $|t|=c'\geq c$.

\section{Effectiveness}\label{sec:effective}
	
	Finally, we show how Theorem \ref{thm:main} follows from Theorem \ref{thm:types-ok}, i.e., how given a \lY-term $P$ of complexity $m$ we can check whether $\Gammae\vdash P:\hat\rho_m\triangleright c$ can be derived for arbitrarily large $c$.
	We say that two type judgments are equivalent if they differ only in the value of the flag counter.
	\label{page:effective-def}
	Let us consider a set $\Dd$ of all derivations of $\Gammae\vdash P:\hat\rho_m\triangleright c$ in which on each branch (i.e., each root-leaf path) there are at most three type judgments from every equivalence class,
	and among premisses of each \AppR rule there is at most one type judgment from every equivalence class.
	These derivations use only type judgments $\Gamma\vdash Q:\hat\tau\triangleright d$ with $Q$ being a subterm of $P$ and with $\Gamma(x)\neq\emptyset$ only for variables $x$ appearing in $P$.
	Since a finite \lY-term, even when seen as an infinitary $\lambda$-term, has only finitely many subterms,
	this introduces a common bound on the height of all derivations in $\Dd$, and on their degree (i.e., on the maximal number of premisses of a rule).
	It follows that there are only finitely many derivations in $\Dd$, and thus we can compute all of them.
	
	We claim that $\Gammae\vdash P:\hat\rho_m\triangleright c$ can be derived for arbitrarily large $c$ if and only if in $\Dd$ there is a derivation in which on some branch
	there are two equivalent type judgments with different values of the flag counter (and the latter condition can be easily checked).
	Indeed, having such a derivation, we can repeat its fragment between the two equivalent type judgments,
	obtaining derivations of $\Gammae\vdash P:\hat\rho_m\triangleright c$ with arbitrarily large $c$.
	We use here an additivity property of our type system: if out of $\Gamma\vdash Q:\hat\tau\triangleright d$ we can derive $\Gamma'\vdash Q':\hat\tau'\triangleright d'$, 
	then out of $\Gamma\vdash Q:\hat\tau\triangleright d+k$	we can derive $\Gamma'\vdash Q':\hat\tau'\triangleright d'+k$, for every $k\geq-d$.
	Conversely, take a derivation of $\Gammae\vdash P:\hat\rho_m\triangleright c$ for some large enough $c$.
	Suppose that some of its \AppR rules uses two equivalent premisses.
	These premisses concern the argument subterm, which is of smaller order than the operator subterm, and thus of order at most $m-1$. 
	The set of marker orders in these premisses has to be empty, as the sets of marker orders from all premisses have to be disjoint.
	Thus, by Lemma \ref{lem:zero-when-no-marker}, the flag counter in our two premisses is $0$.
	\label{page:effective-narrow}
	In consequence, we can remove one of the premisses, without changing anything in the remaining part of the derivation, even the flag counters.
	In this way we clean the whole derivation, so that at the end among premisses of each \AppR rule there is at most one type judgment from every equivalence class.
	The degree is now bounded, and at each node the flag counter grows only by a constant above the sum of flag counters from the children.
	Thus, if $c$ is large enough, we can find on some branch two equivalent type judgments with different values of the flag counter.
	Then, for some pairs of equivalent type judgments, we remove the part of the derivation between these type judgments (and we adopt appropriately the flag counters in the remaining part).
	It it not difficult to perform this cleaning so that the resulting derivation will be in $\Dd$, and simultaneously on some branch there will remain two equivalent type judgments with different values of the flag counter.

\section{Conclusions}
	In this paper, we have shown an approach for expressing quantitative properties of B\"ohm trees using an intersection type system, on the example of the finiteness problem.
	It is an ongoing work to apply this approach to the diagonal problem, which should give a better complexity than that of the algorithm from \cite{downward-closure}.
	Another ongoing work is to obtain an algorithm for model checking B\"ohm trees with respect to the Weak MSO+U logic \cite{DBLP:conf/stacs/BojanczykT12}.
	This logic extends Weak MSO by a new quantifier U, expressing that a subformula holds for arbitrarily large finite sets.
	Furthermore, it seems feasible that our methods may help in proving a pumping lemma for nondeterministic HORSes.

\bibliographystyle{eptcs}
\bibliography{bib}

\begin{thebibliography}{10}
\providecommand{\bibitemdeclare}[2]{}
\providecommand{\surnamestart}{}
\providecommand{\surnameend}{}
\providecommand{\urlprefix}{Available at }
\providecommand{\url}[1]{\texttt{#1}}
\providecommand{\href}[2]{\texttt{#2}}
\providecommand{\urlalt}[2]{\href{#1}{#2}}
\providecommand{\doi}[1]{doi:\urlalt{http://dx.doi.org/#1}{#1}}
\providecommand{\bibinfo}[2]{#2}

\bibitemdeclare{article}{achim-pumping}
\bibitem{achim-pumping}
\bibinfo{author}{Achim \surnamestart Blumensath\surnameend}
  (\bibinfo{year}{2008}): \emph{\bibinfo{title}{On the Structure of Graphs in
  the {C}aucal Hierarchy}}.
\newblock {\sl \bibinfo{journal}{Theor. Comput. Sci.}}
  \bibinfo{volume}{400}(\bibinfo{number}{1-3}), pp. \bibinfo{pages}{19--45},
  \doi{10.1016/j.tcs.2008.01.053}.

\bibitemdeclare{article}{achim-erratum}
\bibitem{achim-erratum}
\bibinfo{author}{Achim \surnamestart Blumensath\surnameend}
  (\bibinfo{year}{2013}): \emph{\bibinfo{title}{Erratum to "On the Structure of
  Graphs in the {C}aucal Hierarchy" [Theoret. Comput. Sci 400 {(2008)}
  19-45]}}.
\newblock {\sl \bibinfo{journal}{Theor. Comput. Sci.}} \bibinfo{volume}{475},
  pp. \bibinfo{pages}{126--127}, \doi{10.1016/j.tcs.2012.12.044}.

\bibitemdeclare{inproceedings}{DBLP:conf/stacs/BojanczykT12}
\bibitem{DBLP:conf/stacs/BojanczykT12}
\bibinfo{author}{Miko{\l}aj \surnamestart Boja{\'{n}}czyk\surnameend} \&
  \bibinfo{author}{Szymon \surnamestart Toru{\'{n}}czyk\surnameend}
  (\bibinfo{year}{2012}): \emph{\bibinfo{title}{Weak {MSO+U} over Infinite
  Trees}}.
\newblock In \bibinfo{editor}{Christoph \surnamestart D{\"{u}}rr\surnameend} \&
  \bibinfo{editor}{Thomas \surnamestart Wilke\surnameend}, editors: {\sl
  \bibinfo{booktitle}{29th International Symposium on Theoretical Aspects of
  Computer Science, {STACS} 2012, February 29th - March 3rd, 2012, Paris,
  France}}, {\sl \bibinfo{series}{LIPIcs}}~\bibinfo{volume}{14},
  \bibinfo{publisher}{Schloss Dagstuhl - Leibniz-Zentrum fuer Informatik}, pp.
  \bibinfo{pages}{648--660}, \doi{10.4230/LIPIcs.STACS.2012.648}.

\bibitemdeclare{inproceedings}{saturation}
\bibitem{saturation}
\bibinfo{author}{Christopher~H. \surnamestart Broadbent\surnameend},
  \bibinfo{author}{Arnaud \surnamestart Carayol\surnameend},
  \bibinfo{author}{Matthew \surnamestart Hague\surnameend} \&
  \bibinfo{author}{Olivier \surnamestart Serre\surnameend}
  (\bibinfo{year}{2012}): \emph{\bibinfo{title}{A Saturation Method for
  Collapsible Pushdown Systems}}.
\newblock In \bibinfo{editor}{Artur \surnamestart Czumaj\surnameend},
  \bibinfo{editor}{Kurt \surnamestart Mehlhorn\surnameend},
  \bibinfo{editor}{Andrew~M. \surnamestart Pitts\surnameend} \&
  \bibinfo{editor}{Roger \surnamestart Wattenhofer\surnameend}, editors: {\sl
  \bibinfo{booktitle}{Automata, Languages, and Programming - 39th International
  Colloquium, {ICALP} 2012, Warwick, UK, July 9-13, 2012, Proceedings, Part
  {II}}}, {\sl \bibinfo{series}{Lecture Notes in Computer Science}}
  \bibinfo{volume}{7392}, \bibinfo{publisher}{Springer}, pp.
  \bibinfo{pages}{165--176}, \doi{10.1007/978-3-642-31585-5\_18}.

\bibitemdeclare{inproceedings}{DBLP:conf/csl/BroadbentK13}
\bibitem{DBLP:conf/csl/BroadbentK13}
\bibinfo{author}{Christopher~H. \surnamestart Broadbent\surnameend} \&
  \bibinfo{author}{Naoki \surnamestart Kobayashi\surnameend}
  (\bibinfo{year}{2013}): \emph{\bibinfo{title}{Saturation-Based Model Checking
  of Higher-Order Recursion Schemes}}.
\newblock In \bibinfo{editor}{Simona Ronchi~Della \surnamestart
  Rocca\surnameend}, editor: {\sl \bibinfo{booktitle}{Computer Science Logic
  2013 {(CSL} 2013), {CSL} 2013, September 2-5, 2013, Torino, Italy}}, {\sl
  \bibinfo{series}{LIPIcs}}~\bibinfo{volume}{23}, \bibinfo{publisher}{Schloss
  Dagstuhl - Leibniz-Zentrum fuer Informatik}, pp. \bibinfo{pages}{129--148},
  \doi{10.4230/LIPIcs.CSL.2013.129}.

\bibitemdeclare{inproceedings}{downward-closure}
\bibitem{downward-closure}
\bibinfo{author}{Lorenzo \surnamestart Clemente\surnameend},
  \bibinfo{author}{Pawel \surnamestart Parys\surnameend},
  \bibinfo{author}{Sylvain \surnamestart Salvati\surnameend} \&
  \bibinfo{author}{Igor \surnamestart Walukiewicz\surnameend}
  (\bibinfo{year}{2016}): \emph{\bibinfo{title}{The Diagonal Problem for
  Higher-Order Recursion Schemes is Decidable}}.
\newblock In \bibinfo{editor}{Martin \surnamestart Grohe\surnameend},
  \bibinfo{editor}{Eric \surnamestart Koskinen\surnameend} \&
  \bibinfo{editor}{Natarajan \surnamestart Shankar\surnameend}, editors: {\sl
  \bibinfo{booktitle}{Proceedings of the 31st Annual {ACM/IEEE} Symposium on
  Logic in Computer Science, {LICS} '16, New York, NY, USA, July 5-8, 2016}},
  \bibinfo{publisher}{{ACM}}, pp. \bibinfo{pages}{96--105},
  \doi{10.1145/2933575.2934527}.

\bibitemdeclare{inproceedings}{DBLP:conf/fsttcs/ClementePSW15}
\bibitem{DBLP:conf/fsttcs/ClementePSW15}
\bibinfo{author}{Lorenzo \surnamestart Clemente\surnameend},
  \bibinfo{author}{Paweł \surnamestart Parys\surnameend},
  \bibinfo{author}{Sylvain \surnamestart Salvati\surnameend} \&
  \bibinfo{author}{Igor \surnamestart Walukiewicz\surnameend}
  (\bibinfo{year}{2015}): \emph{\bibinfo{title}{Ordered Tree-Pushdown
  Systems}}.
\newblock In \bibinfo{editor}{Prahladh \surnamestart Harsha\surnameend} \&
  \bibinfo{editor}{G.~\surnamestart Ramalingam\surnameend}, editors: {\sl
  \bibinfo{booktitle}{35th {IARCS} Annual Conference on Foundation of Software
  Technology and Theoretical Computer Science, {FSTTCS} 2015, December 16-18,
  2015, Bangalore, India}}, {\sl
  \bibinfo{series}{LIPIcs}}~\bibinfo{volume}{45}, \bibinfo{publisher}{Schloss
  Dagstuhl - Leibniz-Zentrum fuer Informatik}, pp. \bibinfo{pages}{163--177},
  \doi{10.4230/LIPIcs.FSTTCS.2015.163}.

\bibitemdeclare{article}{Damm82}
\bibitem{Damm82}
\bibinfo{author}{Werner \surnamestart Damm\surnameend} (\bibinfo{year}{1982}):
  \emph{\bibinfo{title}{The {IO-} and OI-Hierarchies}}.
\newblock {\sl \bibinfo{journal}{Theor. Comput. Sci.}} \bibinfo{volume}{20},
  pp. \bibinfo{pages}{95--207}, \doi{10.1016/0304-3975(82)90009-3}.

\bibitemdeclare{inproceedings}{diagonal-safe}
\bibitem{diagonal-safe}
\bibinfo{author}{Matthew \surnamestart Hague\surnameend},
  \bibinfo{author}{Jonathan \surnamestart Kochems\surnameend} \&
  \bibinfo{author}{C.{-}H.~Luke \surnamestart Ong\surnameend}
  (\bibinfo{year}{2016}): \emph{\bibinfo{title}{Unboundedness and Downward
  Closures of Higher-Order Pushdown Automata}}.
\newblock In \bibinfo{editor}{Rastislav \surnamestart Bod{\'{\i}}k\surnameend}
  \& \bibinfo{editor}{Rupak \surnamestart Majumdar\surnameend}, editors: {\sl
  \bibinfo{booktitle}{Proceedings of the 43rd Annual {ACM} {SIGPLAN-SIGACT}
  Symposium on Principles of Programming Languages, {POPL} 2016, St.
  Petersburg, FL, USA, January 20 - 22, 2016}}, \bibinfo{publisher}{{ACM}}, pp.
  \bibinfo{pages}{151--163}, \doi{10.1145/2837614.2837627}.

\bibitemdeclare{inproceedings}{collapsible}
\bibitem{collapsible}
\bibinfo{author}{Matthew \surnamestart Hague\surnameend},
  \bibinfo{author}{Andrzej~S. \surnamestart Murawski\surnameend},
  \bibinfo{author}{C.{-}H.~Luke \surnamestart Ong\surnameend} \&
  \bibinfo{author}{Olivier \surnamestart Serre\surnameend}
  (\bibinfo{year}{2008}): \emph{\bibinfo{title}{Collapsible Pushdown Automata
  and Recursion Schemes}}.
\newblock In: {\sl \bibinfo{booktitle}{Proceedings of the Twenty-Third Annual
  {IEEE} Symposium on Logic in Computer Science, {LICS} 2008, 24-27 June 2008,
  Pittsburgh, PA, {USA}}}, \bibinfo{publisher}{{IEEE} Computer Society}, pp.
  \bibinfo{pages}{452--461}, \doi{10.1109/LICS.2008.34}.

\bibitemdeclare{inproceedings}{Kar-Par-pumping}
\bibitem{Kar-Par-pumping}
\bibinfo{author}{Alexander \surnamestart Kartzow\surnameend} \&
  \bibinfo{author}{Paweł \surnamestart Parys\surnameend}
  (\bibinfo{year}{2012}): \emph{\bibinfo{title}{Strictness of the Collapsible
  Pushdown Hierarchy}}.
\newblock In \bibinfo{editor}{Branislav \surnamestart Rovan\surnameend},
  \bibinfo{editor}{Vladimiro \surnamestart Sassone\surnameend} \&
  \bibinfo{editor}{Peter \surnamestart Widmayer\surnameend}, editors: {\sl
  \bibinfo{booktitle}{Mathematical Foundations of Computer Science 2012 - 37th
  International Symposium, {MFCS} 2012, Bratislava, Slovakia, August 27-31,
  2012. Proceedings}}, {\sl \bibinfo{series}{Lecture Notes in Computer
  Science}} \bibinfo{volume}{7464}, \bibinfo{publisher}{Springer}, pp.
  \bibinfo{pages}{566--577}, \doi{10.1007/978-3-642-32589-2\_50}.

\bibitemdeclare{inproceedings}{KNU-hopda}
\bibitem{KNU-hopda}
\bibinfo{author}{Teodor \surnamestart Knapik\surnameend},
  \bibinfo{author}{Damian \surnamestart Niwiński\surnameend} \&
  \bibinfo{author}{Paweł \surnamestart Urzyczyn\surnameend}
  (\bibinfo{year}{2002}): \emph{\bibinfo{title}{Higher-Order Pushdown Trees Are
  Easy}}.
\newblock In \bibinfo{editor}{Mogens \surnamestart Nielsen\surnameend} \&
  \bibinfo{editor}{Uffe \surnamestart Engberg\surnameend}, editors: {\sl
  \bibinfo{booktitle}{Foundations of Software Science and Computation
  Structures, 5th International Conference, {FOSSACS} 2002. Held as Part of the
  Joint European Conferences on Theory and Practice of Software, {ETAPS} 2002
  Grenoble, France, April 8-12, 2002, Proceedings}}, {\sl
  \bibinfo{series}{Lecture Notes in Computer Science}} \bibinfo{volume}{2303},
  \bibinfo{publisher}{Springer}, pp. \bibinfo{pages}{205--222},
  \doi{10.1007/3-540-45931-6\_15}.

\bibitemdeclare{inproceedings}{kobayashi2009types-popl}
\bibitem{kobayashi2009types-popl}
\bibinfo{author}{Naoki \surnamestart Kobayashi\surnameend}
  (\bibinfo{year}{2009}): \emph{\bibinfo{title}{Types and Higher-Order
  Recursion Schemes for Verification of Higher-Order Programs}}.
\newblock In \bibinfo{editor}{Zhong \surnamestart Shao\surnameend} \&
  \bibinfo{editor}{Benjamin~C. \surnamestart Pierce\surnameend}, editors: {\sl
  \bibinfo{booktitle}{Proceedings of the 36th {ACM} {SIGPLAN-SIGACT} Symposium
  on Principles of Programming Languages, {POPL} 2009, Savannah, GA, USA,
  January 21-23, 2009}}, \bibinfo{publisher}{{ACM}}, pp.
  \bibinfo{pages}{416--428}, \doi{10.1145/1480881.1480933}.

\bibitemdeclare{inproceedings}{koba-pumping}
\bibitem{koba-pumping}
\bibinfo{author}{Naoki \surnamestart Kobayashi\surnameend}
  (\bibinfo{year}{2013}): \emph{\bibinfo{title}{Pumping by Typing}}.
\newblock In: {\sl \bibinfo{booktitle}{28th Annual {ACM/IEEE} Symposium on
  Logic in Computer Science, {LICS} 2013, New Orleans, LA, USA, June 25-28,
  2013}}, \bibinfo{publisher}{{IEEE} Computer Society}, pp.
  \bibinfo{pages}{398--407}, \doi{10.1109/LICS.2013.46}.

\bibitemdeclare{inproceedings}{context-sensitive-2}
\bibitem{context-sensitive-2}
\bibinfo{author}{Naoki \surnamestart Kobayashi\surnameend},
  \bibinfo{author}{Kazuhiro \surnamestart Inaba\surnameend} \&
  \bibinfo{author}{Takeshi \surnamestart Tsukada\surnameend}
  (\bibinfo{year}{2014}): \emph{\bibinfo{title}{Unsafe Order-2 Tree Languages
  Are Context-Sensitive}}.
\newblock In \bibinfo{editor}{Anca \surnamestart Muscholl\surnameend}, editor:
  {\sl \bibinfo{booktitle}{Foundations of Software Science and Computation
  Structures - 17th International Conference, {FOSSACS} 2014, Held as Part of
  the European Joint Conferences on Theory and Practice of Software, {ETAPS}
  2014, Grenoble, France, April 5-13, 2014, Proceedings}}, {\sl
  \bibinfo{series}{Lecture Notes in Computer Science}} \bibinfo{volume}{8412},
  \bibinfo{publisher}{Springer}, pp. \bibinfo{pages}{149--163},
  \doi{10.1007/978-3-642-54830-7\_10}.

\bibitemdeclare{inproceedings}{kobayashiOng2009type}
\bibitem{kobayashiOng2009type}
\bibinfo{author}{Naoki \surnamestart Kobayashi\surnameend} \&
  \bibinfo{author}{C.{-}H.~Luke \surnamestart Ong\surnameend}
  (\bibinfo{year}{2009}): \emph{\bibinfo{title}{A Type System Equivalent to the
  Modal Mu-Calculus Model Checking of Higher-Order Recursion Schemes}}.
\newblock In: {\sl \bibinfo{booktitle}{Proceedings of the 24th Annual {IEEE}
  Symposium on Logic in Computer Science, {LICS} 2009, 11-14 August 2009, Los
  Angeles, CA, {USA}}}, \bibinfo{publisher}{{IEEE} Computer Society}, pp.
  \bibinfo{pages}{179--188}, \doi{10.1109/LICS.2009.29}.

\bibitemdeclare{inproceedings}{Ong-hoschemes}
\bibitem{Ong-hoschemes}
\bibinfo{author}{C.{-}H.~Luke \surnamestart Ong\surnameend}
  (\bibinfo{year}{2006}): \emph{\bibinfo{title}{On Model-Checking Trees
  Generated by Higher-Order Recursion Schemes}}.
\newblock In: {\sl \bibinfo{booktitle}{21th {IEEE} Symposium on Logic in
  Computer Science {(LICS} 2006), 12-15 August 2006, Seattle, WA, USA,
  Proceedings}}, \bibinfo{publisher}{{IEEE} Computer Society}, pp.
  \bibinfo{pages}{81--90}, \doi{10.1109/LICS.2006.38}.

\bibitemdeclare{article}{jfp-numerals}
\bibitem{jfp-numerals}
\bibinfo{author}{Pawel \surnamestart Parys\surnameend} (\bibinfo{year}{2016}):
  \emph{\bibinfo{title}{A Characterization of Lambda-Terms Transforming
  Numerals}}.
\newblock {\sl \bibinfo{journal}{J. Funct. Program.}} \bibinfo{volume}{26}, p.
  \bibinfo{pages}{e12}, \doi{10.1017/S0956796816000113}.

\bibitemdeclare{inproceedings}{ho-new}
\bibitem{ho-new}
\bibinfo{author}{Paweł \surnamestart Parys\surnameend} (\bibinfo{year}{2012}):
  \emph{\bibinfo{title}{On the Significance of the Collapse Operation}}.
\newblock In: {\sl \bibinfo{booktitle}{Proceedings of the 27th Annual {IEEE}
  Symposium on Logic in Computer Science, {LICS} 2012, Dubrovnik, Croatia, June
  25-28, 2012}}, \bibinfo{publisher}{{IEEE} Computer Society}, pp.
  \bibinfo{pages}{521--530}, \doi{10.1109/LICS.2012.62}.

\bibitemdeclare{article}{itrs-arxiv}
\bibitem{itrs-arxiv}
\bibinfo{author}{Paweł \surnamestart Parys\surnameend} (\bibinfo{year}{2017}):
  \emph{\bibinfo{title}{Intersection Types and Counting}}.
\newblock {\sl \bibinfo{journal}{CoRR}} \bibinfo{volume}{abs/1701.05303v1}.
\newblock \urlprefix\url{http://arxiv.org/abs/1701.05303v1}.

\bibitemdeclare{inproceedings}{DBLP:conf/popl/RamsayNO14}
\bibitem{DBLP:conf/popl/RamsayNO14}
\bibinfo{author}{Steven~J. \surnamestart Ramsay\surnameend},
  \bibinfo{author}{Robin~P. \surnamestart Neatherway\surnameend} \&
  \bibinfo{author}{C.{-}H.~Luke \surnamestart Ong\surnameend}
  (\bibinfo{year}{2014}): \emph{\bibinfo{title}{A Type-Directed Abstraction
  Refinement Approach to Higher-Order Model Checking}}.
\newblock In \bibinfo{editor}{Suresh \surnamestart Jagannathan\surnameend} \&
  \bibinfo{editor}{Peter \surnamestart Sewell\surnameend}, editors: {\sl
  \bibinfo{booktitle}{The 41st Annual {ACM} {SIGPLAN-SIGACT} Symposium on
  Principles of Programming Languages, {POPL} '14, San Diego, CA, USA, January
  20-21, 2014}}, \bibinfo{publisher}{{ACM}}, pp. \bibinfo{pages}{61--72},
  \doi{10.1145/2535838.2535873}.

\bibitemdeclare{article}{MSC:9613738}
\bibitem{MSC:9613738}
\bibinfo{author}{Sylvain \surnamestart Salvati\surnameend} \&
  \bibinfo{author}{Igor \surnamestart Walukiewicz\surnameend}
  (\bibinfo{year}{2016}): \emph{\bibinfo{title}{Simply Typed Fixpoint Calculus
  and Collapsible Pushdown Automata}}.
\newblock {\sl \bibinfo{journal}{Mathematical Structures in Computer Science}}
  \bibinfo{volume}{26}(\bibinfo{number}{7}), pp. \bibinfo{pages}{1304--1350},
  \doi{10.1017/S0960129514000590}.

\end{thebibliography}

\end{document}